\pdfoutput=1
\documentclass[
groupedaddress,
nofootinbib,longbibliography,
twocolumn,
notitlepage,
aps]
{revtex4-2}
\usepackage{amsmath,amsfonts,amssymb,amsthm,amsbsy,mathtools}
\usepackage{graphicx}
\usepackage{dcolumn}
\usepackage{bm}
\usepackage{bbold}
\usepackage{mathtools}
\usepackage{xcolor}
\usepackage{tikz}
\usepackage{tikz-cd}
\usepackage{comment}
\usepackage{subfigure}
\usepackage{comment}
\usepackage{soul}
\usepackage{framed}
\usepackage{mdframed}
\usepackage{csquotes}
\usepackage{appendix}
\usepackage{hyperref}

\newtheorem{theorem}{Theorem}
\newtheorem{corollary}{Corollary}
\newtheorem{lemma}{Lemma}

\usepackage{natbib}
\usepackage[normalem]{ulem}
\usepackage{xcolor}
\usepackage{xspace}
\usepackage{ifthen}

\newcommand{\ket}[1]{{|#1\rangle}}
\newcommand{\bra}[1]{{\langle#1|}}

\newcommand{\Tr}{\operatorname{Tr}}

\newcommand{\showcomments}{true}

\newcommand{\acd}[1]%
{\ifthenelse{\equal{\showcomments}{true}}{{\color{magenta}{#1}}}{\xspace}}%

\begin{document}

\author{Anne-Catherine de la Hamette}
\affiliation{Institute for Theoretical Studies, ETH Zürich, 8006 Zürich, Switzerland}

\date{March 24, 2026}
\title{Observer-Dependent Entropy and Diagonal Rényi Invariants \\ in Quantum Reference Frames}

\begin{abstract}
Quantum reference frames provide a relational description of multipartite quantum systems in which physical states and observables are defined relative to quantum observers. Yet different observers can assign different entropies to the same system, raising the question of how such observer-dependence is constrained. We identify a family of frame-independent diagonal Rényi entropies for arbitrary subsystems, yielding a generalized multipartite coherence–entanglement tradeoff. For ideal frames, the observer-dependence of subsystem entropy admits an exact decomposition into a sum of single-frame coherences and inter-frame correlations; for non-ideal frames, it is instead bounded by the dimension of an effective relational Hilbert space determined by the representation structure of the frames. Our results place quantitative limits on how much quantum observers can disagree about subsystem entropy, with potential implications for observer-dependent entropy assignments in gravitational settings.
\end{abstract}
\maketitle

\section{Introduction}

In physics, properties assigned to physical systems generally depend on the reference frame used to describe them. In quantum theory, this feature becomes particularly subtle, especially when reference frames themselves are treated as quantum systems. Recent developments in quantum reference frames (QRFs) have shown that coherence, entanglement, and even subsystem structure are in general frame-dependent \cite{Giacomini_2017_covariance, Vanrietvelde_2018a,delaHamette_2020,Savi_2020,hoehn2021quantum, delaHamette_2021_perspectiveneutral, Hoehn_2023_subsystems, Kabel_2024, cepollaro2024}. At the interface of quantum information and gravity, such perspective dependence has been suggested to play a role in understanding observer-dependent entropy assignments \cite{deVuyst_2024, De_Vuyst_2025}, for instance in the description of quantum fields and black hole radiation by different observers.

For ideal QRFs -- frames carrying the regular representation of a symmetry group $G$ -- changes of perspective act as permutations in the natural group-label basis. One consequence of this permutation structure is a coherence–entanglement tradeoff between two observers describing the same physical situation \cite{cepollaro2024}. It is known that this tradeoff fails for non-ideal frames, but the breakdown has not yet been quantitatively characterized.

In this work, we considerably extend the previously identified coherence–entanglement tradeoff for ideal QRFs. First, we derive a family of diagonal Rényi invariants for subsystems containing a reference frame, yielding a generalized \emph{multipartite} coherence–entanglement tradeoff. Second, for ideal frames, we show that the observer-dependence of the entropy assigned to a system of interest admits an exact decomposition into contributions from single-frame coherences and multipartite correlations between frames. Third, building on the representation-theoretic analysis of non-ideal frames in the perspective-neutral framework \cite{delaHamette_2021_perspectiveneutral,delaHamette_2021_entanglement}, we derive a state-independent bound on the observer-dependent entropy difference in terms of effective relational Hilbert space dimensions determined by the representation multiplicities of the frames. Together, these results place structural limits on observer-dependent entropy assignments and clarify how frame non-ideality constrains the redistribution of information between observers.

\section{Permutation structure and invariants for ideal frames}

We consider a collection of $N$ quantum reference frames  $R_1,\dots,R_N$ and an additional system $S$. All systems carry unitary representations of a compact group $G$ equipped with normalized Haar measure $dg$.

In this section we restrict to \emph{ideal frames}, meaning that each
frame $R_i$ carries the left-regular representation of $G$
on $\mathcal H_{R_i}=L^2(G)$, with group-label basis $\{\ket{g}\}_{g\in G}$ satisfying $\langle g|g'\rangle=\delta(g,g')$, while the system $S$ carries an arbitrary unitary representation of $G$.

The kinematical Hilbert space is $\mathcal H_{\rm kin}= \mathcal H_{R_1}\otimes \cdots \otimes \mathcal H_{R_N}\otimes \mathcal H_S$, equipped with the diagonal action $U(g):=U_{R_1}(g)\otimes\cdots\otimes U_{R_N}(g)\otimes U_S(g)$. The physical Hilbert space is defined as the invariant subspace
\begin{equation}
\mathcal H_{\rm phys}=\{\ket{\psi}\in\mathcal H_{\rm kin}\;|\; U(g)\ket{\psi}=\ket{\psi}\ \forall g\in G\},
\end{equation}
equivalently obtained via coherent group averaging, with projector $\Pi_{\rm phys}=\int_G dg\, U(g)$. Throughout, we restrict to pure physical states  $\rho_{\rm phys}=\ket{\psi}\!\bra{\psi}$ with $\ket{\psi}\in\mathcal H_{\rm phys}$ as naturally obtained in the perspective-neutral framework by coherent group averaging of pure kinematical states.

For each frame $R_i$, the (pure) \textit{relational state} in the perspective of $R_i$ is obtained via the \textit{reduction map} $\mathcal R_i:\mathcal H_{\rm phys}\to\mathcal H_{SR_{\bar i}|R_i}$, defined by conditioning on the identity element $e\in G$ \cite{delaHamette_2021_perspectiveneutral},
\begin{equation}
\rho^{(R_i)}_{SR_{\bar i}}:=\mathcal R_i(\rho_{\rm phys})=N_i\,{}_{R_i}\langle e|\rho_{\rm phys}|e\rangle_{R_i},
\end{equation}
where $S R_{\bar i}$ denotes the complement of $R_i$, and $N_i$ a normalization constant with $ N_i=d_{R_i}$ when $d_{R_i}=\dim\mathcal H_{R_i}<\infty$ . For any $g\in G$, conditioning on $\ket{g}_{R_i}$ instead yields $U_{SR_{\bar i}}(g)\,\rho^{(R_i)}_{SR_{\bar i}}\,U_{SR_{\bar i}}(g)^\dagger$, so different choices of $g$ are related by the induced group action on $SR_{\bar i}$. 

We next recall the change of perspective between two ideal frames. The QRF transformation from $R_i$ to $R_j$ can be written explicitly in the group-label basis as\footnote{While we formulate the results for compact groups in general, in the special case of finite groups Haar integrals reduce to $\frac{1}{|G|}\sum_{g\in G}$ and probability densities become finite probability vectors on $G$.} 
\begin{equation}
U_{i\to j}=\int_G dg\;{}_{R_i}\ket{g^{-1}}\!\bra{g}_{R_j}
\;\otimes\;U_{S R_{\overline{ij}}}(g^{-1}), \label{eq:qrf-change}
\end{equation}
where $U_{S R_{\overline{ij}}}(g)$ denotes the action of $G$ on all systems excluding frames $R_i$ and $R_j$ \cite{delaHamette_2020}. The relational states are related by
\begin{equation}
    \rho^{(R_j)}_{SR_{\bar j}} = U_{i\to j}\rho^{(R_i)}_{SR_{\bar i}}U_{i\to j}^\dagger.
\end{equation}
One can equally obtain this frame change by composing the reduction maps: $\rho^{(R_j)} =\mathcal{R}_j\mathcal{R}_i^\dagger\rho^{(R_i)}\mathcal{R}_i\mathcal{R}_j^\dagger$.

Before stating the first lemma, let us fix some notation that will be used throughout the paper. 
Let $\Delta$ denote the full dephasing map in the group-label basis on any subsystem of frames and/or the system $S$. 
Diagonal states in this basis correspond to classical probability distributions over the basis labels.

Fix a perspective $R_i$, and choose any subsystem $X\subseteq SR_{\bar i}$ that contains at least one other frame $R_j$. The reduced relational state of $X$ relative to $R_i$ is $\rho^{(R_i)}_X:=\Tr_{\overline X}[\rho^{(R_i)}]$ with $\overline X:=SR_{\bar i}\setminus X$ the complement of $X$. 
After changing perspective to $R_j$, the subsystem that is described relationally is $Y:=(X\cup\{R_i\})\setminus\{R_j\}\subseteq SR_{\bar j}$, obtained from $X$ by exchanging $R_j$ for $R_i$. Its complement is $\overline Y:=SR_{\bar j}\setminus Y$ and the corresponding reduced relational state of $Y$ relative to $R_j$ is $\rho^{(R_j)}_Y:=\Tr_{\overline Y}[\rho^{(R_j)}]=\Tr_{\overline Y}[U_{i\to j}\rho^{(R_i)}U_{i\to j}^\dagger]$. 

It was shown in \cite{cepollaro2024} that the full frame change between two ideal frames acts as a permutation in the group-label basis. Let us now state the following for the case of arbitrarily many frames on reduced subsystems.
\begin{lemma}
\label{lem:perm-strong}
There exists a unitary $P_{X\to Y}:\mathcal H_X\to\mathcal H_Y$ that acts as a permutation in the group-label basis such that
\begin{equation}
\Delta\rho^{(R_j)}_Y=P_{X\to Y}\,\Delta\rho^{(R_i)}_X\,P_{X\to Y}^\dagger. \label{eq:perm-reduced}
\end{equation}
\end{lemma}
The detailed proof is given in Supplemental Material~\ref{app:proof-perm-strong}. In the group-label basis, the change of frame from $R_i$ to $R_j$ acts by relabeling basis vectors. Although this does not, in general, imply that the reduced states $\rho_X^{(R_i)}$ and $\rho_Y^{(R_j)}$ are related by a permutation, it does so after full dephasing in the group-label basis. After dephasing, the reduced states are diagonal, and the frame change simply permutes the retained labels. Hence the dephased reduced states are related by the permutation $P_{X\to Y}$.

Lemma~\ref{lem:perm-strong} motivates considering quantities that depend only on the diagonal entries of reduced states in the group-label basis. Applying full dephasing and evaluating the Rényi-$\alpha$ entropy 
\begin{equation}
    S_\alpha(\rho):=\frac{1}{1-\alpha}\log\Tr(\rho^\alpha), \quad \alpha>0,\ \alpha\neq1,
\end{equation}
with the von Neumann entropy recovered in the limit $\alpha\to1$, yields a family of invariants. 
Take the subsystems $X$ and $Y$ as defined above.

\begin{theorem}
\label{thm:F1-strong}
For any $\alpha>0$ for which the entropies are finite,
\begin{equation}
S_\alpha\!\Big(\Delta\rho^{(R_i)}_X\Big)=S_\alpha\!\Big(\Delta\rho^{(R_j)}_Y\Big).
\end{equation}
\end{theorem}

\begin{proof}
By Lemma~\ref{lem:perm-strong}, the dephased reduced states $\Delta\rho^{(R_i)}_X$ and $\Delta\rho^{(R_j)}_Y$ are related by a permutation $P_{X \to Y}$. Since $S_\alpha$ is invariant under unitary conjugation, the diagonal Rényi entropies coincide.
\end{proof}

The diagonal Rényi invariant admits a natural decomposition into entanglement entropy and coherence contributions. For any subsystem $X\subseteq SR_{\bar i}$ one has
\begin{equation}
S_\alpha\!\big(\Delta\rho^{(R_i)}_X\big) = S_\alpha\!\big(\rho^{(R_i)}_X\big) + C_\alpha\!\big(\rho^{(R_i)}_X\big),
\end{equation}
where $C_\alpha(\sigma):=S_\alpha(\Delta\sigma)-S_\alpha(\sigma)$ denotes the \textit{Rényi-$\alpha$ relative entropy of coherence} \cite{Baumgratz_2014,Shao_2017,Vershynina_2022}. Since the relational state $\rho^{(R_i)}_{SR_{\bar i}}$ is pure, $S_\alpha(\rho^{(R_i)}_X)=S_\alpha(\rho^{(R_i)}_{\overline X})$, with $\overline X=SR_{\bar i}\setminus X$. The invariant therefore takes the form
\begin{equation}
S_\alpha(\Delta\rho^{(R_i)}_X) =S_\alpha(\rho^{(R_i)}_{\overline X})+C_\alpha(\rho^{(R_i)}_X).
\label{eq:tradeoff-general}
\end{equation}
This expresses the diagonal invariant in terms of the entanglement across the bipartition $X|\overline X$ and coherence of $X$ in the group-label basis.

\begin{figure}
    \centering
    \includegraphics[width=0.7\linewidth]{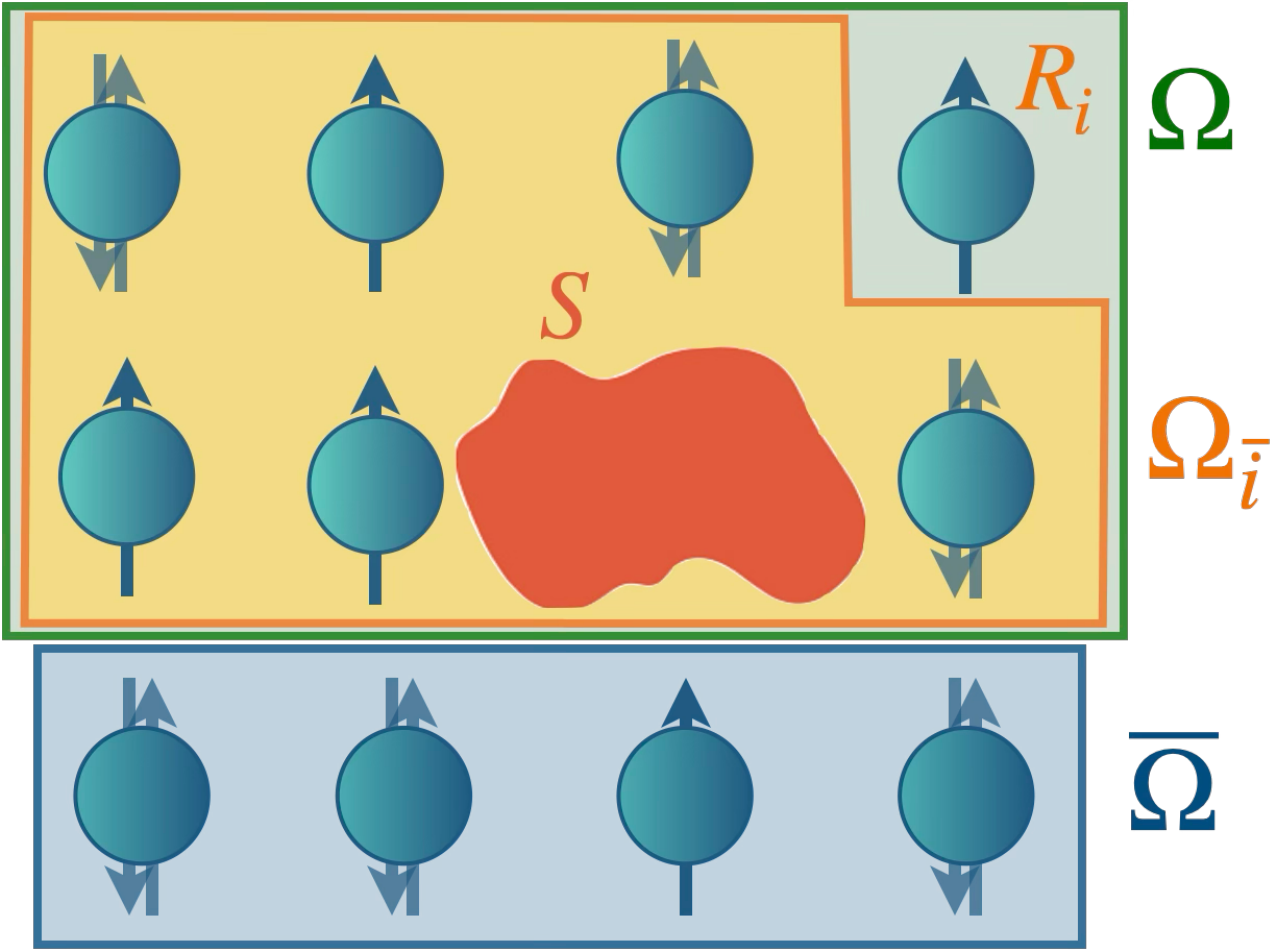}
    \caption{For a collection of quantum reference frames $R_1,\dots,R_N$ (shown here as spins) and a system $S$, choose a subset $\Omega$ containing at least one frame. For any frame $R_i\in\Omega$, let $\Omega_{\bar i}:=\Omega\setminus\{R_i\}$ and $\overline{\Omega}$ denote the complement of $\Omega$. The sum $S_\alpha(\rho^{(R_i)}_{\overline{\Omega}})+C_\alpha(\rho^{(R_i)}_{\Omega_{\bar i}})$ is the same for all choices of $R_i\in\Omega$, expressing agreement among frames in $\Omega$ on this entanglement–coherence invariant.}
    \label{fig:invariant}
\end{figure}

\begin{corollary} \label{cor:multi-tradeoff}
Fix any subset $\Omega\subseteq SR_1\cdots R_N$ containing at least one frame and define $\Omega_{\bar i}:=\Omega\setminus\{R_i\}$ for each $R_i\in\Omega$. Then, for any $\alpha>0$ for which the entropies are finite,
\begin{equation}
S_\alpha\!\big(\rho^{(R_i)}_{\overline{\Omega}}\big)+C_\alpha\!\big(\rho^{(R_i)}_{\Omega_{\bar i}}\big)=\mathrm{const.}\qquad\forall R_i\in\Omega .
\end{equation}
\end{corollary}

\begin{proof}
This follows directly from Theorem~\ref{thm:F1-strong} together with the decomposition \eqref{eq:tradeoff-general}. Note that $\Omega=\{ R_i \} \cup X$ and $X=\Omega_{\overline i}$.
\end{proof}
Thus, all ideal frames in any chosen subset $\Omega$ agree on the sum of Rényi-$\alpha$ entanglement entropy and Rényi-$\alpha$ relative entropy of coherence. The result is illustrated in Fig.~\ref{fig:invariant}.

Several notable cases follow directly. If $X=R_j$, one obtains $S_\alpha(\Delta\rho^{(R_i)}_{R_j}) = S_\alpha(\Delta\rho^{(R_j)}_{R_i})$, which for $\alpha=1$ reproduces the coherence–entanglement tradeoff of \cite{cepollaro2024}, there derived for two ideal frames. Choosing instead $X=S R_{\bar i}$ yields a global relational invariant, $S_\alpha(\Delta\rho^{(R_i)}_{S R_{\bar i}}) = S_\alpha(\Delta\rho^{(R_j)}_{S R_{\bar j}})$ for \emph{all} ideal frames $R_i,R_j$. In this sense, the diagonal Rényi invariants define quantities on which a whole class of ideal observers agree.

%%%%%%%%%%%%%%%%%%%%%%%%%%%%%%%%%%%%%%%%%

The diagonal invariants admit a direct formulation on the physical Hilbert space via the \emph{trivialization map} $\mathcal T_{i}$ \cite{delaHamette_2021_perspectiveneutral}. This map is an isometry $\mathcal T_i:\mathcal H_{\rm phys}\to \mathcal H_{R_i}\otimes \mathcal H_{SR_{\bar i}}$ acting on a pure physical state $\ket{\psi}\in\mathcal H_{\rm phys}$ as
\begin{equation}
\mathcal T_i\ket{\psi}=\ket{\theta}_{R_i}\otimes\ket{\psi(e)}_{SR_{\bar i}},
\end{equation}
where $\ket{\theta}_{R_i}$ is a fixed disentangled pure state of the frame $R_i$ and $\ket{\psi(e)}_{SR_{\bar i}}$ is the relational state in the perspective of $R_i$. Tracing out the trivialized frame recovers the relational state \cite{delaHamette_2021_perspectiveneutral},
\begin{equation}
\Tr_{R_i}\!\big[\mathcal T_i\,\rho_{\rm phys}\,\mathcal T_i^\dagger\big]=\ket{\psi(e)}\!\bra{\psi(e)}_{SR_{\bar i}}=\rho^{(R_i)}_{SR_{\bar i}}. \label{eq:phys-rel-state}
\end{equation}
For any subsystem $X\subseteq SR_{\bar i}$ and corresponding reduced relational state $\rho^{(R_i)}_X:=\Tr_{\overline{X}}[\rho^{(R_i)}_{SR_{\bar i}}]$, Eq.~\eqref{eq:phys-rel-state} implies
\begin{equation}
S_\alpha\!\big(\Delta\,\rho^{(R_i)}_X\big)=S_\alpha\!\Big(\Delta\,\Tr_{R_i\overline{X}}\!\big[\mathcal T_i\,\rho_{\rm phys}\,\mathcal T_i^\dagger\big]
\Big),
\end{equation}
for all $\alpha>0$ for which the entropies are finite. By Theorem~\ref{thm:F1-strong}, the diagonal Rényi invariants are therefore determined by $\rho_{\rm phys}$ alone and are independent of the choice of ideal frame within $\Omega=\{ R_i \} \cup X$.

%%%%%%%%%%%%%%%%%%%%%%%%%%%%%%%%%%%%%%%%%

The diagonal invariants can equivalently be evaluated from any other ideal frame not contained in the chosen subset  $X$. Let $R_k$ be a frame with $R_k\notin X$, and denote by $U_{k\to i}$ the ideal QRF change from the perspective of $R_k$ to that of $R_i$, as in Eq.~\eqref{eq:qrf-change}. Since the full frame change $U_{k\to i}$ relabels the group-label basis by a permutation, the dephased reduced states obtained after tracing out $\overline X$ are related by the corresponding permutation on $X$. Hence, for every subsystem $X\subseteq SR_{\bar i}$,
\begin{equation}
\!\!\!\!  S_\alpha\!\Big(\Delta\,\Tr_{\overline X}[\rho^{(R_i)}] \Big) = S_\alpha\!\Big( \Delta\,\Tr_{\overline X} \big(U_{k\to i}\rho^{(R_k)}U_{k\to i}^\dagger \big) \Big).
\end{equation}

The invariant from Theorem \ref{thm:F1-strong} thus admits three equivalent realizations: in any of the frames contained in $\Omega$, as a canonical quantity on the physical Hilbert space, or in the perspective of any external ideal frame, followed by the appropriate frame change. For ideal frames these descriptions coincide exactly.

%%%%%%%%%%%%%%%%%%%%%%%%%%%%%%%%%%%%%%%%%
The preceding results single out the diagonal Rényi entropies as natural QRF invariants. However, the permutation structure of Lemma~\ref{lem:perm-strong} is stronger: it implies that the diagonal part of the reduced state in the group-label basis is preserved up to relabeling. Consequently, \textit{any}  \enquote{classical} observable that is diagonal in this basis has identical statistics in all ideal frames that agree on the subsystem under consideration. More precisely, let $A_X$ be any operator diagonal in the group-label basis of a subsystem $X\subseteq SR_{\bar i}$ that contains a frame $R_j$. Then the dephased reduced states $\Delta\rho^{(R_i)}_X$ and $\Delta\rho^{(R_j)}_Y$ are related by a permutation $P_{X\to Y}$, and hence, defining $A_Y:=P_{X\to Y}A_XP_{X\to Y}^\dagger$, all moments coincide: 
\begin{align}
    \Tr\!\big(\rho^{(R_i)}_X A_X^n\big)=\Tr\!\big(\rho^{(R_j)}_Y A_Y^n\big)\quad \forall n\in\mathbb N.
\end{align}

\section{Observer-dependent entropy of $S$}
While the diagonal Rényi entropies above are invariant under ideal frame changes, the entropy assigned to the system $S$ itself generally depends on the chosen frame.  
For two ideal frames $R_i$ and $R_j$, define the entropy difference
\begin{equation}
\Delta S^{(i,j)}(S) :=\big| S(\rho^{(R_i)}_S)-S(\rho^{(R_j)}_S)\big|,
\end{equation}
where $S$  denotes the von Neumann entropy ($\alpha \to 1)$.

Let $C(\sigma):=S(\Delta\sigma)-S(\sigma)$ denote the relative entropy
of coherence for $\alpha \to 1$ in the group-label basis \cite{Baumgratz_2014}. Then the observer-dependent
entropy difference admits the following exact expression.

\begin{theorem}
\label{thm:pairwise-bound}
For every pair of ideal frames $(R_i,R_j)$ for which the entropies are finite,
\begin{align}
&\Delta S^{(i,j)}(S) = |C(\rho^{(R_j)}_{R_{\bar j}})-C(\rho^{(R_i)}_{R_{\bar i}})| \nonumber\\
& =|\sum_{k \neq j} C(\rho_{R_k}^{(R_j)})-\sum_{k \neq i} C(\rho_{R_k}^{(R_i)})+ \Gamma^{(R_j)}-\Gamma^{(R_i)}|.
\end{align}
\end{theorem}
\noindent Here $\Gamma^{(R_j)}$ quantifies correlations between frames that are destroyed by dephasing. More precisely, writing $I(A_1...A_n):=\sum_j S(A_j)-S(A_1...A_n)$ for the \emph{multipartite total correlation} \cite{Groisman_2005,Modi_2012}, one defines
\begin{align}
\Gamma^{(R_j)}:=I(R_{\bar j})_{\rho^{(R_j)}_{R_{\bar j}}}-I(R_{\bar j})_{\Delta\rho^{(R_j)}_{R_{\bar j}}}.
\end{align}
The proof of Theorem \ref{thm:pairwise-bound} is given in Supplemental Material~\ref{app:proof-pairwise-bound}. 

The observer-dependence of the entropy of $S$ is thus exactly determined by the difference between the total single-frame coherence assigned by two observers and the amount of inter-frame correlation destroyed by dephasing in their respective perspectives. The latter quantifies how much of the coherence is stored locally in the individual frames and how much is encoded non-locally in correlations among them, showing that entropy differences arise from a redistribution of coherence among the reference frames.

Maximizing the entropy difference over all frames yields
\begin{align}
\Delta S_{\mathrm{max}}(S)\;:=&\max_i S(\rho^{(R_i)}_S)-\min_j S(\rho^{(R_j)}_S)\; \\
=& \max_i C(\rho^{(R_i)}_{R_{\bar i}})-\min_j C(\rho^{(R_j)}_{R_{\bar j}}).
\end{align}
Thus the maximal disagreement between observers about the entropy of $S$ occurs between frames assigning the largest and smallest coherence to the remaining frames. The maximal entropy difference across ideal observers is therefore fully determined by how coherence is redistributed among the reference frames under ideal QRF transformations. The proof is given in Supplemental Material~\ref{app:proof-pairwise-bound}.

\section{Non-ideal frames}
The results above rely crucially on the permutation structure of ideal frame changes. In particular, they generalize the coherence-entanglement tradeoff of \cite{cepollaro2024}, which was shown to hold exactly for ideal frames and to be violated for non-ideal frames.
We now investigate how frame non-ideality constrains the observer-dependence of the entropy assigned to $S$, using the representation-theoretic description of non-ideal frames developed in \cite{delaHamette_2021_perspectiveneutral,delaHamette_2021_entanglement} in the perspective-neutral approach. As we show below, for non-ideal frames, the entropy difference is no longer fixed exactly, but is bounded by the size of the relational state space available to the frames, as determined by the representation content of the frames and the system.

A non-ideal frame $R$ carries a general unitary representation $U_R$ of $G$, which decomposes into irreducible representations as
\begin{equation}
\!\! \mathcal H_R\cong\bigoplus_{q\in\widehat G}\mathbb C^{m_q^R}\otimes \mathcal{N}_q ,\quad U_R(g)=\bigoplus_{q \in\widehat G} I_{m_q^R}\otimes D_q(g),
\end{equation}
where $\widehat G$ labels the irreducible unitary representations of $G$, $D_q$ is the irrep of dimension $d_q$ on $\mathcal N_q$, and $m_q^R$ its multiplicity. Ideal frames correspond to the regular representation, $m_q^R=d_q$.

Following \cite{delaHamette_2021_perspectiveneutral,delaHamette_2021_entanglement}, each (complete) non-ideal frame $R$ is equipped with a covariant coherent-state POVM generated from a seed state $|\phi\rangle\in\mathcal H_R$, with orbit $|\phi(g)\rangle:=U_R(g)|\phi\rangle$. For the orbits to give rise to a resolution of the identity, the multiplicities must satisfy $m_q^R\leq d_q$ for every irrep $q$. Hence a non-ideal frame can carry at most as many copies of each irrep as an ideal frame \cite{delaHamette_2021_perspectiveneutral}.
%Conditioning a joint state $\rho_{RS}$ on the POVM element labelled by $g$ yields a relational state $\rho_S^{(R)}(g)$ of the remaining system relative to the frame orientation $g$.

Consider now two (complete) non-ideal frames $R_1$ and $R_2$ of finite Hilbert space dimensions $d_{R_1}$ and $d_{R_2}$, and a system of interest $S$ of finite dimension $d_S$. Conditioning a pure physical state $\ket{\psi}\in\mathcal H_{\rm phys}$ on $R_1$ being in the coherent state $\ket{\phi(g)}$ yields the pure relational state $\rho_{SR_2}^{(R_1)}(g)$. For $g=e$, we simply write $\rho_{SR_2}^{(R_1)}$.

This immediately gives a first bound on $\Delta S_\alpha(S):=\Delta S_\alpha^{(1,2)}(S)$, for any $\alpha>0$ for which the entropy is finite. Since the relational state is pure, we have $S_\alpha(\rho_S^{(R_1)})=S_\alpha(\rho_{R_2}^{(R_1)})$, and similarly for $R_2$, and thus
\begin{align}
   \Delta S_\alpha(S)&= |S_\alpha(\rho_{R_2}^{(R_1)})-S_\alpha(\rho_{R_1}^{(R_2)})|\nonumber\\&\leq \mathrm{max}\{\log d_{R_2},\log d_{R_1}\}\\
   &= \log \max\left\{ \sum_b m_b^{R_2}d_b,\sum_b m_b^{R_1}d_b\right\}.\nonumber
\end{align}
Here we assume that $S$ is at least as large as either frame, $d_S\ge d_{R_1},d_{R_2}$; otherwise one also has the trivial bound $\Delta S_\alpha(S)\le \log d_S$.

For non-ideal frames, we further bound $S_\alpha(\rho_{R_j}^{(R_i)})$ using $S_\alpha(\rho)\le \log \operatorname{rank}\rho$. We therefore define the \textit{effective relational Hilbert space dimension} $d_{\mathrm{eff}}(R_j|R_i)$ as the maximal support dimension of $\rho_{R_j}^{(R_i)}$ over all physical states (see Supplemental Material~\ref{app: proof-non-ideal-bound}). It follows that $\operatorname{rank}\rho_{R_2}^{(R_1)} \le d_{\mathrm{eff}}(R_2|R_1)$ and thus
\begin{equation}
S_\alpha(\rho_{R_2}^{(R_1)}) \le \log d_{\mathrm{eff}}(R_2|R_1),
\end{equation}
and similarly with $R_1$ and $R_2$ interchanged.

\begin{theorem} \label{thm:non-ideal-bound}
For two non-ideal frames $R_1$ and $R_2$, the Rényi-$\alpha$ entropy difference assigned to $S$ satisfies
\begin{align}
    \Delta S_\alpha(S) \leq \log \max \{ d_{\mathrm{eff}}(R_1|R_2),d_{\mathrm{eff}}(R_2|R_1) \} 
\end{align}
where 
\begin{align}
  d_{\mathrm{eff}}(R_i\,|\,R_j)=\sum_a d_a\,\min\!\Bigl(m_a^{R_i},\sum_{b,c} m_b^{R_j}m_c^{S}N_{bc}^{\,\bar a}\Bigr)  
\end{align}
with $i\neq j \in \{1,2\}$.
\end{theorem}
\noindent The proof can be found in Supplemental Material \ref{app: proof-non-ideal-bound}.

This bound shows that frame non-ideality limits how much two observers can disagree about the entropy of $S$. The maximal disagreement is set by the effective relational Hilbert space available to the frames, as determined by their representation content and that of the system $S$.

Well-defined QRFs give rise to a resolution of the identity and thus satisfy $m_q^R\le d_q$, with equality for ideal frames, i.e.~for the regular representation \cite{delaHamette_2021_perspectiveneutral}. In the ideal case, the effective dimensions $d_{\mathrm{eff}}(R_i|R_j)$ attain their maximal values (equal to $|G|$ for finite groups). Thus, ideal frames yield the largest possible value of the state-independent bound in Theorem~\ref{thm:non-ideal-bound}. In that regime, however, Theorem~\ref{thm:pairwise-bound} provides an exact characterization of the entropy difference.

For non-ideal frames, the strict inequality $m_q^R<d_q$ in at least one sector reduces the size of the relational state space and hence the effective dimensions $d_{\mathrm{eff}}(R_i|R_j)$, leading to a tighter bound on $\Delta S_\alpha(S)$. A simple example is a non-ideal $U(1)$ clock with a finite energy cutoff, carrying only charge sectors $|n|\le \Lambda$. Relative to the ideal case, such a frame lacks high-charge irreps and therefore supports fewer relational degrees of freedom, reducing $d_{\mathrm{eff}}$ and the maximal entropy difference between observers.

\section{Conclusions}
We have derived a family of diagonal invariants for ideal quantum reference frames that extend the previously known coherence–entanglement tradeoff \cite{cepollaro2024}. We generalized it from the case of two frames and one system to arbitrarily many frames, with arbitrary choices of subsets, and from the von Neumann entropy to Rényi-$\alpha$ entropies. Concretely, all ideal frames in a chosen subset agree on the sum of the coherence of that subsystem and the entanglement entropy of its complement. These invariants arise from a subsystem permutation structure of ideal frame changes and constrain how information can be redistributed between observers. We further showed that the observer-dependent entropy assigned to a system of interest admits an exact decomposition into a sum of single-frame coherences and inter-frame correlations.

For non-ideal frames, exact invariance no longer holds. Instead, we derived a state-independent bound on the observer-dependence of the entropy assigned to $S$ in terms of the effective relational Hilbert space accessible to the frames. This bound depends only on the representation-theoretic multiplicities carried by the frames and the system. Since non-ideal frames carry fewer copies of some irreps than ideal frames, the corresponding reduced relational states are supported on smaller effective relational state spaces. The dimension of this support bounds the entropy of the state and, by purity of the relational state, also the entropy assigned to $S$. The maximal possible disagreement in the entropy assigned to $S$ is therefore reduced.

Reference frames may be interpreted operationally as physical rods, clocks, or general measurement devices carried by observers \cite{Loveridge_2016,Loveridge_2017, Giacomini_2017_covariance, Vanrietvelde_2018a, Vanrietvelde_2018b, delaHamette_2020,Loveridge_2020,Castro_Ruiz_2021, delaHamette_2021_perspectiveneutral,Glowacki_2023,Carette_2023,Wang_2023,Chataignier_2024}. In gravitational settings, entropy assignments to quantum fields in curved spacetime, and related gravitational entropy constructions, can depend on the observer’s description \cite{deVuyst_2024,De_Vuyst_2025}. In several recent approaches, non-ideal quantum reference clocks play a crucial role in rendering such entropy computations finite \cite{Chandrasekaran_2022,Jensen_2023,Fewster_2024, deVuyst_2024,Ahmad_2025,De_Vuyst_2025}. Taking the reference frames described in this work as quantum clocks carried by observers, our results provide kinematic constraints on how much two quantum observers can disagree about the entropy of another system due solely to the quantum properties of their reference frames. In particular, our bound shows that this observer-dependence is limited by the size of the relational state space available to the frames. This can have interesting implications for the observer-dependence of gravitational entropies, including entanglement entropies of Hawking radiation. Our results show that such discrepancies are governed by the coherence and correlations carried by the reference systems and, for non-ideal frames, are further limited by their representation-theoretic properties. This suggests that observer-dependence in gravitational entropy calculations can be traced, at least in part, to the coherence, correlations, and non-ideality of the quantum clocks used to define time.

\vspace{0.5 cm}
\paragraph*{Acknowledgments.}
I thank Andrea Di Biagio and Thomas Galley for comments on an earlier draft of this work, and Joe Renes for insightful discussions on entropies. I further thank Andrea Di Biagio for very valuable and encouraging discussions, and for pointing me towards a possible invariant on the level of the physical state space.
I acknowledge support from the Institute for Theoretical Studies at ETH Zurich through a Junior Research Fellowship.

\nocite{apsrev42Control}
\bibliography{bibliography}
\bibliographystyle{apsrev4-2.bst}

%\newpage
\onecolumngrid
\appendix
\section{Proof of Lemma \ref{lem:perm-strong}} \label{app:proof-perm-strong}

Fix two ideal frames $R_i$ and $R_j$ and divide the remaining systems $SR_{\overline{ij}}$ into the two sets $K$ and $T$, where $K$ denotes the subsystems that are kept and $T$ the subsystems that are later traced out. In the group-label basis, the QRF transformation $U_{i\to j}$ in Eq.~\eqref{eq:qrf-change} acts by a relabeling of basis vectors:
\begin{align}
    U_{i\to j}\ket{g}_{R_j}\ket{\mathbf k}_{K}\ket{\mathbf t}_{T}=\ket{g^{-1}}_{R_i}\ket{g^{-1}\mathbf k}_{K}\ket{g^{-1}\mathbf t}_{T},
\end{align}
where $\mathbf{k}=(k_{h_1},\dots,k_{h_m})$ and $\mathbf{t}=(t_{h'_1},\dots,t_{h'_n})$ denote the basis labels of the remaining subsystems.
 Let
$X:=R_jK,\ Y:=R_iK,$
and define the reduced states
\begin{align}
  \rho_X^{(R_i)}:=\Tr_T \rho^{(R_i)},\qquad
\rho_Y^{(R_j)}:=\Tr_T\!\big(U_{i\to j}\rho^{(R_i)}U_{i\to j}^\dagger\big).  
\end{align}
Let $\Delta$ denote the full dephasing map in the group-label basis. More precisely, one could write $\Delta_X$ and $\Delta_Y$ for the corresponding dephasing maps on $X$ and $Y$, but we simply use $\Delta$ throughout.\\

For notational simplicity, we present the proof for finite groups. For compact groups, the same argument applies formally upon replacing sums by Haar integrals and Kronecker deltas by Dirac distributions.

Write the initial state in the group-label basis of $R_jKT$ as
\begin{align}
    \rho^{(R_i)}=\sum_{g,\mathbf k,\mathbf t}\sum_{g',\mathbf k',\mathbf t'}\rho_{g,\mathbf k,\mathbf t;\,g',\mathbf k',\mathbf t'}\ket{g,\mathbf k,\mathbf t}\!\bra{g',\mathbf k',\mathbf t'}.
\end{align}
Tracing out $T$ gives
\begin{align}
   \rho_X^{(R_i)}=\sum_{g,\mathbf k}\sum_{g',\mathbf k'}\sum_{\mathbf t}\rho_{g,\mathbf k,\mathbf t;\,g',\mathbf k',\mathbf t}\ket{g,\mathbf k}\!\bra{g',\mathbf k'}. 
\end{align}
Applying $\Delta$ removes all off-diagonal terms in the group-label basis, hence
\begin{align} \Delta\!\big(\rho_X^{(R_i)}\big)=\sum_{g,\mathbf k}\sum_{\mathbf t}\rho_{g,\mathbf k,\mathbf t;\,g,\mathbf k,\mathbf t}\ket{g,\mathbf k}\!\bra{g,\mathbf k}. \label{eq: dephased state in X}
\end{align}
Now apply the frame change to the full state:
\begin{align}
    U_{i\to j}\rho^{(R_i)}U_{i\to j}^\dagger=\sum_{g,\mathbf k,\mathbf t}\sum_{g',\mathbf k',\mathbf t'}\rho_{g,\mathbf k,\mathbf t;\,g',\mathbf k',\mathbf t'}\ket{g^{-1},g^{-1}\mathbf k,g^{-1}\mathbf t}\!\bra{g'^{-1},g'^{-1}\mathbf k',g'^{-1}\mathbf t'}.
\end{align}
Tracing out $T$ yields
\begin{align}
    \rho_Y^{(R_j)}=\sum_{g,\mathbf k,\mathbf t}\sum_{g',\mathbf k',\mathbf t'}\rho_{g,\mathbf k,\mathbf t;\,g',\mathbf k',\mathbf t'}\,\delta_{g^{-1}\mathbf t,\;g'^{-1}\mathbf t'}\,\ket{g^{-1},g^{-1}\mathbf k}\!\bra{g'^{-1},g'^{-1}\mathbf k'}.
\end{align}
Dephasing in the group-label basis keeps only the diagonal terms, so necessarily $g=g'$ and $\mathbf k=\mathbf k'$. Since left multiplication by $g^{-1}$ is bijective, the condition
$g^{-1}\mathbf t=g^{-1}\mathbf t'$
then implies $\mathbf t=\mathbf t'$. Therefore
\begin{align}
\Delta\!\big(\rho_Y^{(R_j)}\big)=\sum_{g,\mathbf k}\sum_{\mathbf t}\rho_{g,\mathbf k,\mathbf t;\,g,\mathbf k,\mathbf t}\ket{g^{-1},g^{-1}\mathbf k}\!\bra{g^{-1},g^{-1}\mathbf k}.
\label{eq: dephased state in Y}
\end{align}
Next define $P_{X\to Y}$ on the group-label basis by
\begin{align}
    P_{X\to Y}\ket{g,\mathbf k}:=\ket{g^{-1},g^{-1}\mathbf k}.
\end{align}
This is unitary, because the map $(g,\mathbf k)\mapsto(g^{-1},g^{-1}\mathbf k)$ is a bijection of the group-label basis, hence $P_{X\to Y}$ is a unitary that acts as a permutation in the group-label basis.

Applying $P_{X\to Y}$ to Eq.~\eqref{eq: dephased state in X}, we obtain
\begin{align}
   P_{X\to Y}\,\Delta\!\big(\rho_X^{(R_i)}\big)\,P_{X\to Y}^\dagger=\sum_{g,\mathbf k}\sum_{\mathbf t}\rho_{g,\mathbf k,\mathbf t;\,g,\mathbf k,\mathbf t}\ket{g^{-1},g^{-1}\mathbf k}\!\bra{g^{-1},g^{-1}\mathbf k}, 
\end{align}
which coincides exactly with Eq.~\eqref{eq: dephased state in Y}. Hence
\begin{align}    \Delta\!\big(\rho_Y^{(R_j)}\big)=P_{X\to Y}\,\Delta\!\big(\rho_X^{(R_i)}\big)\,P_{X\to Y}^\dagger.
\end{align}
\qed

\section{Proofs Of Theorem \ref{thm:pairwise-bound} and Maximal Entropy Difference} \label{app:proof-pairwise-bound}

Consider $X=R_{\bar i} \subset SR_{\bar i}$. From Theorem \ref{thm:F1-strong}, we have $S(\Delta \rho_{R_{\bar i}}^{(R_i)})=S(\Delta \rho_{R_{\bar j}}^{(R_j)})$ for all ideal frames $R_i,\ R_j$. Using the decomposition
\begin{align}
    S(\Delta\rho^{(R_i)}_{R_{\bar i}}) = S(\rho^{(R_i)}_S)+C(\rho^{(R_i)}_{R_{\bar i}}),
\end{align}
(because $\rho^{(R_i)}_{SR_{\bar i}}$ is pure), we get
\begin{align}
    S(\rho^{(R_i)}_S)+C(\rho^{(R_i)}_{R_{\bar i}}) = S(\rho^{(R_j)}_S)+C(\rho^{(R_j)}_{R_{\bar j}}).
\end{align}
Therefore,
\begin{align}
    S(\rho^{(R_i)}_S)-S(\rho^{(R_j)}_S)=C(\rho^{(R_j)}_{R_{\bar j}})-C(\rho^{(R_i)}_{R_{\bar i}}),
\end{align}
and hence
\begin{align}
\Delta S^{(R_i,R_j)}(S)=\big| C(\rho^{(R_j)}_{R_{\bar j}})- C(\rho^{(R_i)}_{R_{\bar i}})\big|.
\end{align}

Further, let us show that 
\begin{align} \label{eq:coherence-term-decomp}
    C(\sigma^{(R_j)}) =\sum_{k\neq j} C\!\big(\rho^{(R_j)}_{R_k}\big) +\Big[I(R_{\bar j})_{\sigma^{(R_j)}}-I(R_{\bar j})_{\Delta\sigma^{(R_j)}}\Big],
\end{align}
where $I(R_{\bar j})_{\sigma}:=\sum_{k\neq j} S(\sigma_{R_k})-S(\sigma_{R_{\bar j}})$ and we denote $\sigma^{(R_j)}:=\rho^{(R_j)}_{R_{\bar j}}$.

We can see that this holds as follows. Write
\begin{align} \label{eq: coherence decomp}
    \sum_{k\neq j} C(\sigma_{R_k})=\sum_{k\neq j}\Big(S(\Delta\sigma_{R_k})-S(\sigma_{R_k})\Big).
\end{align}

Since total dephasing is local, we have $(\Delta\sigma)_{R_k}=\Delta(\sigma_{R_k}),$ so $S(\Delta\sigma_{R_k})=S((\Delta\sigma)_{R_k}).$ Hence, 
\begin{align}
    \sum_{k\neq j} C(\sigma_{R_k})=\sum_{k\neq j} S((\Delta\sigma)_{R_k})-\sum_{k\neq j} S(\sigma_{R_k}).
\end{align}

Now use the definition of multipartite total correlation twice:
\begin{align}
    I(R_{\bar j})_{\sigma}=\sum_{k\neq j} S(\sigma_{R_k})-S(\sigma), \qquad I(R_{\bar j})_{\Delta\sigma}=\sum_{k\neq j} S((\Delta\sigma)_{R_k})-S(\Delta\sigma).
\end{align}

Rearranging these two equations and substituting them into Eq.~\eqref{eq: coherence decomp} gives
\begin{align}
    \sum_{k\neq j} C(\sigma_{R_k}) =\big(I(R_{\bar j})_{\Delta\sigma}+S(\Delta\sigma)\big)-\big(I(R_{\bar j})_{\sigma}+S(\sigma)\big).
\end{align}

Using $C(\sigma)=S(\Delta\sigma)-S(\sigma)$, this becomes
\begin{align}
    \sum_{k\neq j} C(\sigma_{R_k})=C(\sigma)-\big(I(R_{\bar j})_{\sigma}-I(R_{\bar j})_{\Delta\sigma}\big),
\end{align}
hence
\begin{align}
C(\sigma)=\sum_{k\neq j} C(\sigma_{R_k})+\Big(I(R_{\bar j})_{\sigma}-I(R_{\bar j})_{\Delta\sigma}\Big).
\end{align}
We have thus shown Eq.~\eqref{eq:coherence-term-decomp}. Then, using the exact tradeoff identity
$\Delta S^{(R_i,R_j)}(S) =\big|C(\rho^{(R_j)}_{R_{\bar j}})-C(\rho^{(R_i)}_{R_{\bar i}})\big|$, the claim of the theorem follows:
\begin{align}
\Delta S^{(R_i,R_j)}(S)=\Bigg|\sum_{k\neq j} C\!\big(\rho^{(R_j)}_{R_k}\big)-\sum_{k\neq i} C\!\big(\rho^{(R_i)}_{R_k}\big)+\Gamma^{(R_j)}-\Gamma^{(R_i)}\Bigg|,
\end{align}
with $\Gamma^{(R_j)}:=I(R_{\bar j})_{\rho^{(R_j)}_{R_{\bar j}}} -I(R_{\bar j})_{\Delta\rho^{(R_j)}_{R_{\bar j}}}. $
\qed

Next, let us prove the relation
\begin{align}
    \Delta S_{\mathrm{max}}(S)\;:=&\max_i S(\rho^{(R_i)}_S)-\min_j S(\rho^{(R_j)}_S)\; \\
=& \max_i C(\rho^{(R_i)}_{R_{\bar i}})-\min_j C(\rho^{(R_j)}_{R_{\bar j}}).
\end{align}

From Theorem \ref{thm:F1-strong}, it follows that for all ideal frames $R_i$, there exists a single scalar $K$ such that
\begin{align}
    S(\rho_S^{(R_i)})=K-C(\rho_{R_{\bar i}}^{(R_i)})\ \forall i.
\end{align}
Then, maximizing the entropy over frames is equivalent to minimizing the coherence:
\begin{align}
    \max_i S(\rho_S^{(R_i)}) &= K - \min_i C(\rho_{R_{\bar i}}^{(R_i)}), \\
    \min_i S(\rho_S^{(R_i)}) &= K - \max_i C(\rho_{R_{\bar i}}^{(R_i)}).
\end{align}
And thus it follows that 
\begin{align}
    \Delta S_{\mathrm{max}}(S) = \max_i C(\rho^{(R_i)}_{R_{\bar i}})-\min_j C(\rho^{(R_j)}_{R_{\bar j}}).
\end{align}\qed

\section{Proof of Theorem \ref{thm:non-ideal-bound}} \label{app: proof-non-ideal-bound}
We bound the entropy difference $\Delta S_\alpha(S)$ by bounding the support dimension of the reduced relational states $\rho_{R_2}^{(R_1)}$, $\rho_{R_1}^{(R_2)}$. Since $S_\alpha(\rho)\le \log \operatorname{rank}\rho$, it suffices to derive a state-independent upper bound on $\operatorname{rank}\rho_{R_j}^{(R_i)}$.

We first consider $\rho_{R_2}^{(R_1)}$. To determine its maximal support dimension, we decompose the complement of $R_2$, namely $R_1S$:
\begin{align}
\mathcal H_{R_1}\otimes \mathcal H_S\cong\bigoplus_q \mathbb C^{\,n_q^{(1S)}}\otimes \mathcal N_q,
\end{align}
with $n_q^{(1S)}=\sum_{a,c} m_a^{R_1}\,m_c^{S}\,N_{ac}^{\,q}$ and $N_{ac}^{\,q}$ the Clebsch-Gordan coefficients.

Now consider a fixed irrep sector $b$ of $R_2$. For a physical invariant state, this sector can only couple to sector $\bar b$ on $R_1S$. The maximal Schmidt rank across the multiplicity spaces in sector $b$ is therefore $\min\!\bigl(m_b^{R_2},\,n_{\bar b}^{(1S)}\bigr)$. Since each such contribution carries an irrep factor of dimension $d_b$, the maximal contribution of sector $b$ to $\operatorname{rank}\rho_{R_2}^{(R_1)}$ is $d_b\,\min\!\bigl(m_b^{R_2},\,n_{\bar b}^{(1S)}\bigr)$.

Summing over all sectors gives the effective dimension of $\rho_{R_2}^{(R_1)}$,
\begin{align}
d_{\mathrm{eff}}(R_2\,|\,R_1)=\sum_b d_b\,\min\!\Bigl(m_b^{R_2},\sum_{a,c} m_a^{R_1}m_c^{S}N_{ac}^{\,\bar b}\Bigr).
\end{align}
Analogously, we find the effective dimension of $\rho_{R_1}^{(R_2)}$,
\begin{align}
d_{\mathrm{eff}}(R_1\,|\,R_2)=\sum_b d_b\,\min\!\Bigl(m_b^{R_1},\sum_{a,c} m_a^{R_2}m_c^{S}N_{ac}^{\,\bar b}\Bigr).
\end{align}

Since $\mathrm{rank}\ \rho_{R_2}^{(R_1)} \leq d_{\mathrm{eff}}(R_2|R_1)$ and $\mathrm{rank}\ \rho_{R_1}^{(R_2)} \leq d_{\mathrm{eff}}(R_1|R_2)$,
as argued in the main text, and $S_\alpha(\rho)\le \log \operatorname{rank}\rho$, the Rényi-$\alpha$ entropies assigned to $S$ from the perspectives of $R_1$ and $R_2$ are bounded by $\log d_{\mathrm{eff}}(R_2|R_1)$ and $\log d_{\mathrm{eff}}(R_1|R_2)$, respectively. Hence,
\begin{align}
    \Delta S_\alpha(S)\le \log \max\!\bigl\{d_{\mathrm{eff}}(R_1\,|\,R_2),\,d_{\mathrm{eff}}(R_2\,|\,R_1)\bigr\}.
\end{align}
\qed
\end{document}